\documentclass{article}

\usepackage{latexsym}
\usepackage{tabularx}
\usepackage{amsfonts}
\usepackage{amsthm}
\usepackage{todonotes}

\usepackage{xcolor}
\usepackage{pst-all} 
\usepackage{pst-poly}
\newpsobject{showgrid}{psgrid}{subgriddiv=1,griddots=10,gridlabels=6pt}
\usepackage{graphicx}
\usepackage{fancybox}
\usepackage{wrapfig}

\usepackage{algorithm}
\usepackage{algpseudocode}

\newcommand{\con}{\wedge} 
\newcommand{\dis}{\vee} 
\newcommand{\alw}{\Box} 
\newcommand{\imp}{\Rightarrow} 
\newcommand{\som}{\Diamond} 

\newtheorem{definition}{Definition}
\newtheorem{theorem}{Theorem}
\newtheorem{corollary}{Corollary}

\title{Towards a Pattern-based Automatic Generation of Logical Specifications for Software Models}

\author{Rados{\l}aw Klimek\\AGH University of Science and Technology,\\ al.\ A.~Mickiewicza 30, 30-059 Krakow, Poland}
\date{2013}

\begin{document}
\maketitle




\begin{abstract}
The work relates to the automatic generation of logical specifications,
considered as sets of temporal logic formulas,
extracted directly from developed software models.
The extraction process is based on the assumption that
the whole developed model is structured using only predefined workflow patterns.
A method of automatic transformation of workflow patterns to
logical specifications is proposed.
Applying the presented concepts enables bridging the gap between
the benefits of deductive reasoning for the correctness verification process and
the difficulties in obtaining complete logical specifications for this process.
\end{abstract}



\section{Introduction}
\label{sec:introduction}

The deductive approach is always an essential part of the human thought process and
might provide the natural support for reasoning about
system correctness and guarantee a rigorous approach in software constructions.
The approach also dominates scientific works and is used intuitively in everyday life.
It enables the analysis of infinite computation sequences.
The need to manually build a systems's logical specification
seems to be a significant obstacle to the practical use of deduction-based verification tools.
Hence, the challenge of automating the process of obtaining logical specifications
seems to be understandable, justified and particularly important.
A large class of software models can be organized and developed using workflow patterns,
e.g.\ business models expressed in BPMN~\cite{Aalst-etal-2003},
or the UML activity diagrams~\cite{Klimek-2013-sefm}.
However, the presented approach is more general and, for example,
business models should not be related directly to the patterns of the work despite
any literal similarity.

The main purpose of the work is to propose a kind of language for
expressing software models and modeling logical specifications.
The main idea is the generation of logical specifications,
a process based on both workflow patterns and their predefined temporal logic formulas
which describe every pattern's behavior.
The work's contribution is a method that
automates the generation of logical specifications.
The generation algorithm for workflow patterns is presented.
The proposed method is characterized by the following advantages:
introducing patterns as primitives to logical modeling,
scaling up to real-world problems,
and logical patterns once defined,
e.g.\ by a logician or a person with good skills in logic,
then widely used,
e.g.\ by analysts and developers with less skills in logic.
All these factors are discussed in the work and summarized in the last section.

The considerations in the work are focused on
the \emph{propositional linear time logic} PLTL,
for which syntax and semantics are defined in~\cite{Emerson-1990}.
It should be pointed that atomic propositions could be identical to
atomic activities, or tasks, of which particular patterns can be constructed.
The work's considerations are justifiably limited to
the minimal temporal logic, e.g.~\cite{vanBenthem-1995}.
Less expressiveness of this logic, compared to more complex ones,
can be compensated or counterbalanced for the fact
that it is much easier to build a deduction engine,
or use the existing valid provers, for this logic.
For example, the deduction engine for the semantic tableaux method for
the minimal temporal logic is relatively easy to build.
Thus, the proposed approach can be verified more quickly than
in the case of more complex logical systems.
Last but not least, the work is focused on
automating the process of generating logical specifications for temporal logic
and minimal temporal logic seems to be the first step in the research,
with more complex logics to follow.

\section{Pattern-oriented logical modeling}
\label{sec:pattern-based-modeling}

A \emph{pattern} is a generic description of structure of some computations,
and does not limit the possibility of modeling arbitrary complex models.
Every workflow pattern can be linked
with a set of temporal logic formulas describing its properties (liveness and safety),
i.e.\ expressing the behavior of tasks or activities included in a pattern.

Let us assume that syntactically correct temporal logic formulas
are already defined~\cite{Emerson-1990}.
\begin{definition}
The \emph{elementary set of formulas}
$pat(a_{1}, \ldots, a_{n})$, or $pat()$,
build over atomic formulas $a_{1}, \ldots, a_{n}$
is a set of temporal logic formulas $f_{1}, \ldots, f_{m}$
such that all formulas are syntactically correct,
i.e.\ $pat()=\{ f_{1}, \ldots, f_{m} \}$.
\end{definition}
A general example for the definition is given as
$pat(a,b,c)=\{a\imp\som b, \alw\neg(b \con\neg c)\}$
which is a two-element set of LTL formulas created over three atomic formulas.

The key notion for the work is a logical pattern considered as a structure.
\begin{definition}
\label{def:logical-pattern}
The \emph{logical pattern} is a structure $Pattern=\langle ini,fin,pat()\rangle $,
where  $ini$ and $fin$ are logical statements of classical propositional logic
describing the logical circumstances of, respectively,
the start and the termination of the whole workflow pattern execution,
and $pat()$ is an elementary set of temporal logic formulas describing
the behavior of a workflow pattern.
\end{definition}
\begin{figure}[htb]
\centering
\scalebox{0.7}{
\begin{pspicture}(6,7) 
\psset{framearc=0,linewidth=1.5pt}
\psset{shadow=false,shadowcolor=gray}
\psset{linecolor=black,nodesep=1pt}
\cnode(2,6){3pt}{p1}
\cnode(4,6){3pt}{p2}
\cnode*(3,4.5){2pt}{p3}
\ncline{->}{p1}{p3}
\ncline{->}{p2}{p3}
\cnode*(3,3){2pt}{p4}
\ncline{->}{p3}{p4}
\nccurve[angleA=0,angleB=0,nodesep=2pt]{->}{p4}{p3}
\cnode*(3,2){2pt}{p5}
\ncline{->}{p4}{p5}
\cnode*(3,1){2pt}{p6}
\ncline{->}{p5}{p6}
\cnode*(2,2){2pt}{p7}
\psline[linewidth=.15cm](2.8,1)(3.2,1)
\ncline{->}{p4}{p7}
\psline[linewidth=.15cm](1.8,2)(2.2,2)
\rput(3,6.8){\rnode{t1}{\textcolor{blue}{ini}}}
\rput(2,.3){\rnode{t2}{\textcolor{blue}{fin}}}
\nccurve[linecolor=blue,angleA=180,angleB=90,linestyle=dotted,linewidth=1pt]{-}{t1}{p1}
\nccurve[linecolor=blue,angleA=0,angleB=90,linestyle=dotted,linewidth=1pt]{-}{t1}{p2}
\nccurve[linecolor=blue,angleA=0,angleB=270,linestyle=dotted,linewidth=1pt]{-}{t2}{p6}
\nccurve[linecolor=blue,angleA=180,angleB=250,linestyle=dotted,linewidth=1pt]{-}{t2}{p7}
\end{pspicture}
}
\caption{The illustration for $ini$- and $fin$-conditions for a pattern}
\label{fig:illustration-ini-fin}
\end{figure}
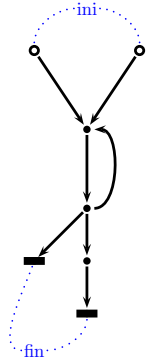
The meaning of $pat()$ does not raise any questions.
$ini$ and $fin$ describe logic circumstances associated with, respectively,
the initiation/opening and termination/closing of a pattern.
$ini$ is satisfied when some initial activity or activities are active.
For example, $a \con b$ for $ini$ means that when the execution of a pattern
is initiated then both activities $a$ and $b$ are satisfied.
$a \dis b$ for $fin$ means that when execution of a pattern is to be terminated,
then activity $a$ or activity $b$ is active.
In other words,
$ini$- and $fin$-expressions (conditions) describe, respectively, the first and
the last activity, or activities, of the pattern which are active.
(Thus, $ini$ and $fin$ should not be confused with the well-known,
pre- and post-conditions, respectively.)
The idea of $ini$ and $fin$ conditions for a pattern is shown in Fig.~\ref{fig:illustration-ini-fin}.
Transitions illustrate flows of control inside a pattern.
$ini$ and $fin$ are expressed in classical propositional logic since they are valuated at a certain time point
as point expressions.

$Pattern.pat$ is a general notation for a set of formulas for a pattern,
and in the case when a particular pattern is considered
it can be written as $Seq.pat$ (for the Sequence pattern),
or, if it does not lead to ambiguity then $Seq$.
Similarly, instead of $Seq.ini$ and $Seq.fin$ can be written $ini$ and $fin$.

Logical consistency is one of the most important properties of logical systems.
A set of logical formulas is \emph{consistent} if it does not contain contradiction.
In other words, logical theory is consistent if and only if (iff)
there is no situation that both $\Phi$ and $\neg\Phi$ are
provable from the axioms and formulas of the theory.
When building logical patterns for workflow patterns,
it is assumed that logical patterns are consistent,
i.e.\ prepared \emph{basic formulas} $pat()=\{ p_{1} , ... , p_{n} \}$ are consistent,
and also all temporal formulas describing behavior/valuation over the $ini$- and $fin$-conditions,
let us call them \emph{transition formulas}, are consistent,
e.g.\ $\som ini$, $ini \imp \som fin$, etc.
\begin{definition}
\label{def:pattern-consistent}
The logical pattern is \emph{consistent} iff
all formulas associated with this pattern,
including formulas for $ini$- and $fin$-conditions,
i.e.\ $pat() \cup \{\som ini, ini \imp \som fin, \som fin, \alw\neg(ini \con fin) \}$
are consistent.
\end{definition}
The discussed transition formulas are important and describe
both safety and liveness aspect of a pattern which is considered as a whole,
i.e.\ in terms of the $ini$- or $fin$-conditions.
In certain situations, if desired, it enables
abstracting from details of a pattern internal behavior and examining it as a whole.
Both $ini$- and $fin$-expressions have equal rights to the pattern representation.
\begin{corollary}
Every pattern contains, and its logical pattern describes, the structure consisting of
at least two activities (or tasks).
\end{corollary}
Instead of a formal proof, note that,
in the general case, e.g.\ for a one-activity pattern,
if the above statement is not satisfied
then the formula $\alw\neg(ini \con fin)$ is also not satisfied.

\begin{algorithm}[htb]
\caption{Labeling logical expressions}
\label{alg:labeled-expression}
{\small
\begin{algorithmic}[1]
\algrenewcommand\algorithmicrequire{\textbf{Input:}}
\algrenewcommand\algorithmicensure{\textbf{Output:}}
\Require logical expression $W_{L}$ (non-empty)
\Ensure labeled logical expression $W_{L}'$
\Statex \Comment{Lexical tokens are processed/scanned one by one from left to right}
\State{l:= 0; pre:= $\epsilon$;}  \Comment{l=current label, pre=previous label}
\For{every $token$ scanned in the whole expression}
\If{token=``(''}
\State{l$++$; replace ``('' by ``(l]'' for the token}
\ElsIf{token=``)''}
\If{pre=``)''} \State l$--$ \EndIf
\State{replace ``)'' by ``[l)'' for the token}
\EndIf
\State{pre:= token}
\EndFor
\end{algorithmic}
}
\end{algorithm}
A symbolic notation allows for literal representation of a structure of
an arbitrary complexity and including nested patterns.
\begin{definition}
\label{def:logical-expression}
The \emph{logical expression} $W_{L}$ is a structure created using the following rules:
\begin{enumerate}
\item every elementary set $pat(a_{i})$,
      where $i>0$ and every $a_{i}$ is an atomic formula,
      is a logical expression,
\item every $pat(A_{i})$, where $i>0$ and every $A_{i}$ is either
      \begin{itemize}
      \item an atomic formula, or
      \item a logical expression $pat()$,
      \end{itemize}
      is also a logical expression.
\end{enumerate}
\end{definition}
\begin{definition}
\label{def:labeled-logical-expression}
The \emph{labeled logical expression} $W_{L}'$ is
received from a logical expression $W_{L}$
introducing (numerical) labels into the logical expression using
the Algorithm~\ref{alg:labeled-expression}.
\end{definition}
The labeled logical expression shows directly the nested structure of patterns.
The function of the \emph{maximum label} $max$ for an expression
returns the maximum (numerical) value of the label placed in the expression.\\
$Seq(1]a,Seq(2]ParalSplit(3]b,c,d[3),Synchron(3]e,f,g[3)[2)[1)=w$ is an example
of the labeled expression which is intuitive in that it shows
the sequence that leads to a parallel split and then synchronization of
some activities and the maximum label is $max(w)=3$.

\begin{figure}[htb]
\centering
{\footnotesize
\begin{minipage}{.4\linewidth}
\begin{verbatim}
Seq(f1,f2):
ini= f1 / fin= f2
f1 => <>f2 / ~f1=>~<>f2
[]~(f1 & f2)
Concur(f1,f2,f3):
ini= f1 / fin= f2 | f3
f1 => <>f2 & <>f3  /  ~f1 => ~<>f2 & ~<>f3
[]~(f1 & (f2 | f3))
Branch(f1,f2,f3):
ini= f1 / fin= (f2 & ~f3) | (~f2 & f3)
f1 => (<>f2 & ~<>f3) | (~<>f2 & <>f3)
~f1 => ~<>(f1 | f2)
[]~(f2 & f3) / []~((f1 & f2) | (f1 & f3))
\end{verbatim}
\end{minipage}
}
\caption{A sample predefined set $P$}
\label{fig:predefined-P}
\end{figure}
Logical properties of any pattern from \emph{predefined temporal patterns}~$\Pi$
are expressed in temporal logic and stored in the predefined set $P$,
c.f.\ example in Fig.~\ref{fig:predefined-P}.
The set is a plain ASCII text/file
to indicate that it can be easy modified using a simple text editor.
Most elements of the $P$~set,
i.e.\ two temporal logic operators, classical logic operators, are not in doubt.
The slash allows to separate formulas placed in a single line.
$f_{1}$, $f_{2}$ etc.\ are atomic formulas
and constitute a kind of formal arguments for a pattern.
Although this sample set $P$ contains only three predefined temporal patterns
$\Pi=\{ Seq, Concur, Branch\}$,
i.e.\ sequential order of activities,
and fork of control to enable concurrent execution of activities,
and conditional execution of activities,
there is no difficulty with defining a set of elementary formulas for other workflow patterns.
The considerations of this work are limited to three patterns to present the main idea of
the work which is the pattern-based generation of logical specifications.

Let us supplement Definition~\ref{def:logical-pattern}, which provides information about
the partial order ($\preceq$) for the set of atomic formulas as arguments of a pattern.
It consists of three subsets which are pairwise disjoint:
the $ini$ arguments (at least one element),
ordinary arguments (may be empty), and
the $fin$ arguments (at least one element).
$f_{i} \preceq f_{j}$ iff $f_{i} \prec f_{j}$ or $f_{i} = f_{j}$.
$f_{i} = f_{j}$ iff both $f_{i}$ and $f_{j}$ belong exclusively to only one of the three subsets.
$f_{i} \prec f_{j}$ iff $f_{i}$ belongs to $ini$ arguments and $f_{j}$ belongs to ordinary arguments, or
$f_{i}$ belongs to ordinary arguments and $f_{j}$ belongs to $fin$ arguments.
All the $ini$ arguments (and no others) form the $ini$-expression.
All the $fin$ arguments (and no others) form the $fin$-expression.
None of the ordinary arguments are included either in the $ini$- or $fin$-expression.

One may need to calculate the consolidated $ini$- and $fin$-expressions to obtain
expressions for complex and nested patterns.
\begin{definition}
\label{def:consolidated-conditions}
Let $w^{c}$ for a logical expression $w$ with the upper index $c=i$ (or $f$, respectively) be
the \emph{consolidated $ini$-expression} (or the \emph{consolidated $fin$-expression}, respectively).
The consolidated expression is calculated using the following (recursive) rules:
\begin{enumerate}
\item if there is no pattern itself in the place of any atomic argument which syntactically belongs to
      the $ini$-expression (or the $fin$-expression, respectively) $w$,
      then $w^{i}$ is equal to $w.ini$ ($w^{f}$ is equal to $w.fin$, respectively),
\item if there is a pattern, say $t()$, in a place of any atomic argument, say $r$, which syntactically belongs to
      the $ini$-expression (or the $fin$-expression, respectively) of $w$,
      then $r$ is replaced by $t^{i}$ (or $t^{f}$, respectively) for every such case.
\end{enumerate}
\end{definition}
Examples of consolidated expressions with regard to the patterns from Fig.~\ref{fig:predefined-P} are given as follows.
For $w=Seq(a,b)$ conditions are $w^{i}=a$ and $w^{f}=b$ (step~1).
For $w=Concur(a,Seq(b,c),d)$ conditions are $w^{i}=a$ (step~1) and
$w^{f}=c \dis d$ (step~2 gives $Seq^{f}=$``$c$'' and step~1 gives ``$\dis d$'', which is consolidated to ``$c \dis d$'').
For $w=Concur(a,Seq(b,Concur(c,d,e)),f)$ conditions are $w^{i}=a$ and $w^{f}=(d \dis e) \dis f$
(step~2 gives $Seq^{f}$ and step~1 gives ``$\dis f$'',
the next step~2 for $Seq^{f}$ leads to $Concur^{f}=$``$(d \dis e)$'',
and after consolidation ``$(d \dis e) \dis f$'' is obtained).

\section{Generation of logical specifications}
\label{sec:generation-specifications}

\begin{figure}[htb]
\centering
\scalebox{0.7}{
\begin{pspicture}(6,5.5) 
\psset{framearc=0}
\psset{shadow=false,shadowcolor=gray}
\psset{linecolor=black}
\cnode(0,4.3){0pt}{g0}
\rput(3.2,4.3){\rnode{g1}{\psframebox
                      {\begin{tabular}{c}
                            \textcolor{black}{\textsc{Patterns}}\\
                            \textcolor{black}{\textsc{scanner}}
                      \end{tabular}}}}
\rput(.2,.5){\rnode{g3}{P}}
\rput(3.2,1.5){\rnode{g2}{\psframebox
                      {\begin{tabular}{c}
                            \textcolor{black}{\textsc{Logical}}\\
                            \textcolor{black}{\textsc{specification}}\\
                            \textcolor{black}{\textsc{generator (${\cal A}$\ref{alg:generating-specification})}}
                      \end{tabular}}}}
\rput(6,1.5){\rnode{g4}{}}

\psset{linewidth=1.5pt}
\psset{shadow=false}
\psset{linecolor=black}
\ncline[angleA=90,angleB=0]{->}{g0}{g1}\Aput{Model\qquad}
\ncline[angleA=0,angleB=0]{->}{g1}{g2}\Aput{$W_{L}$}
\nccurve[angleA=90,angleB=180,nodesep=2pt]{->}{g3}{g2}
\ncline[angleA=90,angleB=0]{->}{g2}{g4}\Aput{$L$}
\psframe[linestyle=dotted,linewidth=1pt](1.1,0.1)(5.3,5.4)
\end{pspicture}
}
\caption{System for generating logical specifications}
\label{fig:formulas-generator}
\end{figure}
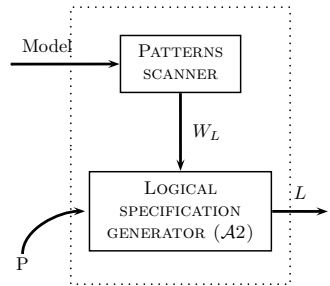
Generating logical specifications requires
an algorithm that converts a logical expression to the target logical specification.
\begin{definition}
\label{def:logical-specification}
\emph{Logical specification} $L$ consists of all formulas derived from
a logical expression $W_{L}$ using the Algorithm~\ref{alg:generating-specification}
(shortly ${\cal A}$\ref{alg:generating-specification}),
i.e.\ $L(W_{L}) = \{f_{i} : i \geq 0 \con f_{i} \in {\cal A}\ref{alg:generating-specification}(W_{L},P)\}$,
where $f_{i}$ is a PLTL formula.
\end{definition}

Generating a logical specification is not a simple summation of
formula collections resulting from the components of a logical expression.
The architecture of the whole proposed system is shown in Fig.~\ref{fig:formulas-generator}.
The generation algorithm~${\cal A}$\ref{alg:generating-specification} has two inputs.
The first is a logical expression $W_{L}$ (Def.~\ref{def:logical-expression})
which is a kind of variable in that it varies for every
(workflow) model.
The second is a predefined set~$P$ (Fig.~\ref{fig:predefined-P})
which is a kind of constant in that it is predefined for
a certain class of software models.
The output of the generation algorithm ${\cal A}$\ref{alg:generating-specification}
is a logical specification $L$ (Def.~\ref{def:logical-specification})
understood as a set of temporal logic formulas.

\begin{algorithm}[htb]
\caption{Generating logical specifications}
\label{alg:generating-specification}
{\small
\begin{algorithmic}[1]
\algrenewcommand\algorithmicrequire{\textbf{Input:}}
\algrenewcommand\algorithmicensure{\textbf{Output:}}
\Require Logical expression $W_{L}$ (non-empty), predefined set $P$ (non-empty)
\Ensure Logical specification $L$
\State $L:=\emptyset$ \Comment{initiating specification}\label{alg:2:atomic-ini}
\For{$l$:= $max(W_{L}')$ \textbf{to} 1}\label{alg:2:atomic-for}
\State p:= getPat($W_{L}',l$); \Comment{get the first (lefmost) pattern with the label $l$}
\Repeat
\If{pattern $p$ consists only atomic formulas}\label{alg:2:atomic-s}
\State{$L := L \cup p.pat()$}
\EndIf\label{alg:2:atomic-e}
\If{any argument of the $p$ is a pattern itself}\label{alg:2:non-atomic-s}
\State{Specification $L'$ for every combination $C_{i=1,\ldots,n}$, i.e.\ $L'(C_{i})$, are}
\State{calculated considering $ini$- and $fin$-expressions for every non-atomic}
\State{arguments and substituting consolidated expressions in places}
\State{of these patterns as arguments, i.e.\ $L := L \cup L'(C_{i})$}
\EndIf\label{alg:2:non-atomic-e}
\State{p:= getPat($W_{L}',l$)}  \Comment{get the next pattern with the label $l$}
\Until{$p$ is empty}
\EndFor
\end{algorithmic}
}
\end{algorithm}
When the Algorithm~\ref{alg:generating-specification} is analyzed,
it seems that lines \ref{alg:2:non-atomic-s} to \ref{alg:2:non-atomic-e} need some clarification.
For example, the consideration of the logical expression $p(a,q(),r(),d)$,
where only $a$ and $d$ are atomic formulas,
leads to the following combinations $C_{i=1,\dots,4}$:
$p(a,q^{i}(),r^{i}(),d)$,
$p(a,q^{i}(),r^{f}(),d)$,
$p(a,q^{f}(),r^{i}(),d)$, and
$p(a,q^{f}(),r^{f}(),d)$.
Let us supplement the ${\cal A}$\ref{alg:generating-specification} by some other examples.
For lines \ref{alg:2:atomic-s}--\ref{alg:2:atomic-e}:
$Branch(a,b,c)$ gives
$L=\{ a \imp (\som b \con \neg\som c) \dis (\neg\som b \con \som c),
\neg a \imp\neg\som (b \dis c),
\alw\neg (b \con c), \alw\neg((a \con b) \dis (a \con c)) \}$.
The example for lines \ref{alg:2:non-atomic-s}--\ref{alg:2:non-atomic-e}:
$Concur(Seq(a,b),c,d)$ leads to
$L = \{ a \imp\som b, \neg a \imp\neg\som b, \alw\neg (a \con b)\} \cup
\{ a \imp \som c \con \som d, \neg a \imp \neg\som c \con \neg\som d, \alw\neg (a \con (c \dis d)) \} \cup
\{ b \imp \som c \con \som d, \neg b \imp \neg\som c \con \neg\som d, \alw\neg (b \con (c \dis d)) \}$,
where the first set results from the nested pattern,
the second results from the use of the $ini$-expression considered as an argument instead of the nested pattern,
and the third results from the use of the $fin$-expression also considered as an argument instead of the nested pattern.
Thus, the final specification is
$L = \{
a \imp\som b, \neg a \imp\neg\som b, \alw\neg (a \con b),
a \imp \som c \con \som d, \neg a \imp \neg\som c \con \neg\som d, \alw\neg (a \con (c \dis d)),
b \imp \som c \con \som d, \neg b \imp \neg\som c \con \neg\som d, \alw\neg (b \con (c \dis d))
\}$.

The basic question is always consistency of logic specifications.
\begin{theorem}
Supposing that every pattern of the $P$ set is non-empty and consistent
and every pair of patterns has disjointed set of atomic formulas
then the logical specification obtained for the ${\cal A}$\ref{alg:generating-specification} algorithm is consistent.
\end{theorem}
\begin{proof}
Let us consider four cases each of which refers to the important points of the Algorithm~\ref{alg:generating-specification}.\\
\textit{Line~\ref{alg:2:atomic-ini}}. Immediately due to the initiation of $L$.\\
\textit{Line~\ref{alg:2:atomic-for}}. Patterns processing order does not affect the consistency of a logical specification
             since it itself does not add new formulas to the logical specification.\\
\textit{Lines~\ref{alg:2:atomic-s}--\ref{alg:2:atomic-e}}. Due to the assumptions of the disjointedness of atomic formulas for patterns
             and the disjointedness of patterns in a logical expression
             (note that any two patterns are either completely disjointed or
             completely contained one in the other),
             and also the assumption of the predefined patterns consistency,
             c.f.\ Definition~\ref{def:pattern-consistent},
             it is not possible to introduce contradictions when adding a new elementary set of formulas.\\
\textit{Lines~\ref{alg:2:non-atomic-s}--\ref{alg:2:non-atomic-e}}. Due to the consistency of transition and basic formulas,
             c.f.\ Definition~\ref{def:pattern-consistent},
             as well as syntax, and properties, and consistency of transition formulas,
             the newly-generated logical specification is consistent,
             it follows from the fact that new temporal formulas for $ini$- and $fin$-conditions
             refer to the consistent transition formulas of a pattern.
\end{proof}

\section{Conclusions}
\label{sec:conclusions}


The proposed method of generating enables scaling up,
that is migration from small problems,
e.g.\ comprising a single UML activity diagram and its workflow,
or a single pool for business processes,
to real-world problems
with more and more diagrams and pools with nesting patterns.
This gives hope for practical use in the case of problems of any size.
Another advantage for the idea of logical patterns is the fact that
they are once well-defined and could be widely and commonly used by an inexperienced user.
Then, it does not necessarily require in-depth knowledge of the complex aspects of temporal logic by an ordinary user.
Introducing patterns as logical primitives foreshadows progress in building logical specifications
that breaks some barriers and obstacles mentioned in the work for practical use of deduction-based formal methods.

Further works may include:
other logical properties of the approach,
more enriched logics,
different workflow patterns, and
different classes of applications.

\bibliographystyle{elsarticle-num}
\bibliography{rk-bib-rk,rk-bib-main}

\end{document}